\documentclass[runningheads,a4paper,envcountsame]{llncs}
\usepackage{graphicx}
\usepackage{amsmath, amssymb, color}

\usepackage{graphicx}
\usepackage{epsfig}
\usepackage{graphics}
\usepackage{array}
\usepackage{amsmath}
\usepackage{amsfonts}
\usepackage{amssymb}

\usepackage{setspace}
\usepackage{fullpage} 

\usepackage{color}

\DeclareMathOperator{\str}{stretch}

\begin{document}
\mainmatter  
\title{Tree $t$-spanners in Outerplanar Graphs via Supply Demand Partition}
\author{N.S.Narayanaswamy \and G.Ramakrishna}
\institute{Department of Computer Science and Engineering\\ Indian Institute of Technology Madras, India.
\\ \email{\{swamy,grama\}@cse.iitm.ac.in}
}

\maketitle
\begin{abstract}
 A \emph{tree $t$-spanner} of an unweighted graph $G$ is a spanning tree $T$ such that for every two vertices their distance in $T$ is at  most $t$ times their distance in $G$.   
Given an unweighted graph $G$ and a positive integer $t$ as input, the \textsc{tree $t$-spanner} problem is to compute a tree $t$-spanner of $G$ if one exists.
This decision problem is known to be NP-complete even in the restricted class of unweighted planar  graphs.  
  We present a linear-time reduction from \textsc{tree $t$-spanner} in outerplanar graphs to the supply-demand tree partition problem. Based on this reduction, we obtain a linear-time algorithm to solve \textsc{tree $t$-spanner} in outerplanar graphs.  Consequently, we show that the minimum value of $t$ for which an input  outerplanar graph on $n$ vertices has a tree $t$-spanner can be found in $O(n \log n)$ time.
\end{abstract}

\section{Introduction}
The area of finding sparse data structures to maintain approximate distance information in a graph is a deep one with far-reaching importance in both engineering (computer networks) and mathematics (metric embeddings).   The sparsest possible data structures that give relevant distance information are spanning trees, and it is natural to approximate distance information within a desired factor, denoted by the parameter $t$ in this paper.  
 A \emph{tree $t$-spanner} of an unweighted graph $G$ is a spanning tree $T$ such that for every two vertices their distance in $T$ is at  most $t$ times their distance in $G$.  The parameter $t$ is called \emph{stretch}, and this notion was introduced by Peleg and Ullman \cite{hyperCubePelegUllman87}.  Further,  a \emph{minimum stretch spanning tree} of $G$ is a tree $t$-spanner of $G$ for the smallest possible value of $t$.
Finding a \emph{minimum stretch spanning tree} is a classical optimization problem in algorithmic graph theory that has several applications in  networks, distributed systems and biology \cite{Awerbuch_broadcastApplications,applicationInBiology,applicationInDistributed}.   The computational questions associated with tree spanners have been referred to as the  \textsc{tree $t$-spanner} problem.   The results in the literature on the \textsc{tree $t$-spanner} problem fall into two classes: when $t$ is a fixed number, and when $t$ is given as part of the input.  In this paper, we present results on the case when $t$ is part of the input, and on finding a minimum stretch spanning tree.   Further, our focus is on outerplanar graphs.
Outerplanar graphs are well-studied due to their rich combinatorial and topological properties.   We use the fact that they have a unique outerplane embedding and that their weak dual is a tree.  
For  problems like vertex cover, on planar graphs polynomial time approximation schemes, based on Baker's technique \cite{bakersTechnique}, are obtained via algorithms on $k$-outerplanar graphs.  These algorithms have an running time that is exponential in $k$, but polynomial in the graph size.  $1$-outerplanar graphs are outerplanar graphs, and we believe that our results are a first step towards such algorithms for \textsc{tree $t$-spanner} on $k$-outerplanar graphs.  We hope that this will eventually lead   to  $o(\log n)$ approximation algorithms for \textsc{tree $t$-spanner} on planar graphs, when $t$ is part of the input.    Tree spanner problems are closely related to the extensively studied cycle basis problems \cite{LiebchenZooOfTreeSpanners}.   For the minimum cycle basis problem on general graphs \cite{Kavitha09cyclebases}, no linear time algorithms are known. 
On outerplanar graphs, a linear time algorithm is known for the minimum cycle basis problem \cite{MCBOuterPlanar98}. More recently,  it has been shown that lex-short cycles form the {\em unique} minimum cycle basis for outerplanar graphs \cite{LiuL10}. 
  Such a result does not hold for the minimum stretch spanning tree which has a corresponding fundamental cycle basis question \cite{LiebchenZooOfTreeSpanners}.   Our study on outerplanar graphs furthers this connection between efficiently optimizing different metrics associated with cycle bases and the graph structure.
%
Outerplanar graphs also appear in several applications including computational drug design, bioinformatics, and telecommunications \cite{frequentMining10}.  For instance, 94.3 $\%$ of elements in the popular NCI data set are outerplanar \cite{frequentMining10}.   The cyclic structure of many organic chemical compounds is represented by cycle bases of the underlying graphs.

\noindent
 {\bf Related Work.}
 For each fixed $t \geq 4$, Cai and Corneil in \cite{Cai95treespanners} showed that \textsc{tree $t$-spanner} is NP-complete.  \textsc{tree $2$-spanner} is  polynomial-time solvable, and the status of  \textsc{tree $3$-spanner}  is open.   The problem has attracted extensive attention in many graph classes.   For planar graphs,  \textsc{tree $3$-spanner} was shown to be polynomial time solvable by Fekete and Kremer \cite{Fekete01spannersinPlanarGraphs}.   For each fixed $t$,  Fomin et al. showed that  the \textsc{tree $t$-spanner} problem restricted to planar graphs and apex-minor free graphs can be solved in polynomial time.  They showed this by designing an FPT algorithm with $t$ as the parameter in planar graphs and apex-minor free graphs \cite{fominFPTspannersPlanar}.  Also, for any fixed $t \geq 4$, the \textsc{tree $t$-spanner} is NP-complete even on chordal graphs and chordal bipartite graphs \cite{Brandstadt04TreeSpannersChordalGraphs,harnessResultChordalBipartite}.
Interval graphs, permutation graphs, regular bipartite graphs and distance hereditary graphs admit a tree 3-spanner \cite{additiveTreeSpannersKratsch,MadanlalSpanners} that can be found in polynomial time. As a consequence, a minimum stretch spanning tree in such graphs can be found in polynomial time due to the fact that \textsc{tree $2$-spanner} is polynomial-time solvable.  
Further, \textsc{tree 4-spanner} in 2-trees is linear-time solvable \cite{tree4SpannersIn2Trees}.

 While the results  in the previous paragraph are for the case when $t$ is fixed, our results are for the case when $t$ is part of the input.
 When $t$ is part of the input, Fekete and Kremer, in \cite{Fekete01spannersinPlanarGraphs}, showed that \textsc{tree $t$-spanner} in planar graphs is NP-complete.
To the best of our knowledge, when $t$ is part of the input, there is no published result on \textsc{tree $t$-spanner} in any subclass of planar graphs, except for grid graphs (see next paragraph).

The minimum stretch spanning tree problem is also referred to as the Minimum Max-Stretch spanning tree in the literature \cite{lognapproxEmekPeleg}.  Peleg and Emek, in \cite{lognapproxEmekPeleg}, present an $O(\log n)$-approximation algorithm to compute a minimum stretch spanning tree, and also show that unless P = NP, the problem cannot be approximated additively by any $o(n)$ factor.  
 A minimum stretch spanning tree in grid graphs can be found by a polynomial-time algorithm \cite{gridSubGridGraphs03}.   Therefore, when $t$ is part of the input, on grid graphs the \textsc{tree $t$-spanner} problem can be solved in polynomial time.


In this work, we reduce \textsc{tree $t$-spanner} in outerplanar graphs to the supply-demand partition problem in trees (defined later). This partition problem was introduced by Ito et. al  and  has applications in power delivery networks \cite{PartitioningTreesSupplyDemand}. 
The supply-demand partition problem is NP-complete even in bipartite graphs and linear time solvable in trees \cite{PartitioningTreesSupplyDemand}. Recently, Kawabata et. al considered edge capacities in the supply-demand partition in trees and designed a linear-time algorithm \cite{Kawabata13}. See \cite{minCostSupplyDemand,GraphsSupplyDemand} for other variations of supply-demand partition.

 \noindent
{\bf Our Results.} 
Our main result is that we present a linear time algorithm for the \textsc{tree $t$-spanner} problem when $t$ is part of the input  in outerplanar graphs.  Thus, in a manner of speaking, we have one more cycle basis question that admits a linear time algorithm on outerplanar graphs.   
As a consequence of our linear time algorithm, it follows that a minimum stretch spanning tree can be obtained in $O(n \log n)$ time in outerplanar graphs.   Our results use properties of the topological structure of outerplanar graph as opposed to another important property that outerplanar graphs are of treewidth two.  Further, for any  $t$, the property that a graph has a tree $t$-spanner  can be described as a formula in Monadic Second Order Logic (MSOL) \cite{fominFPTspannersPlanar}.  It must be mentioned here that since $t$ is part of the input, the treewidth of the graph being two {\em does not}  result in a linear time algorithm by a direct application of the well known Courcelle's theorem \cite{courcelleTheorem}.   The reason is that  the length of the MSOL formula depends on the parameter $t$, which is not {\em fixed} in the \textsc{tree $t$-spanner} problem that we consider.  In other words, the size of the formula is {\em not} a constant.  Actually, it is believed that the length of the smallest MSOL formula would have an exponential tower dependency on $t$ that is of the form, $2^{2^{2^{2^t}}}$, since the MSOL formulation has 4 alternating quantifiers.  See \cite{towersPaper} on the length of the MSOL formulae as a function of the number of alternating quantifiers.    Another important property of the treewidth being 2 is that an optimal tree decomposition can be obtained in linear time using \cite{FindTreeDecomp_Bodlaender96}.   It is not clear how the standard approaches for solving NP-hard problems on graphs of bounded treewidth \cite{niedermeierBook} can be made to work for \textsc{tree $t$-spanner} on outerplanar graphs when $t$ is part of the input.  However, we exploit the topological structure to obtain a linear time algorithm.

(In Section \ref{SectionProperties}) We first present the structure of a canonical tree $t$-spanner in outerplanar graphs that is useful to establish a connection between \textsc{tree $t$-spanner} in outerplanar graphs  and the \textsc{tree S-partition} problem (defined below).
(In Section \ref{treeSpannersSpartitionSection})
We then present a linear-time reduction from \textsc{tree $t$-spanner} in an outerplanar graph $G$ to \textsc{tree S-partition} in a tree $\tilde{T}$, the weak dual of $G$. Further, we present a linear-time algorithm to construct a tree $t$-spanner of $G$ from a tree S-partition of $\tilde{T}$.
(In Section \ref{SupplyDemand_section})
We describe a linear-time reduction from \textsc{tree S-partition}  in $\tilde{T}$ to the well-studied \textsc{supply-demand tree partition} problem \cite{PartitioningTreesSupplyDemand} in a tree $T'$, which is constructed from $\tilde{T}$.
Further, we present a linear-time algorithm to obtain a tree S-partition of $\tilde{T}$ from a supply-demand tree partition of $T'$.
 We use the  linear-time algorithm in \cite{PartitioningTreesSupplyDemand} to decide \textsc{supply-demand tree partition},  and produces a supply-demand tree partition, if one exists.
 
\noindent
We now state our main results formally.

\begin{theorem}
\label{treeSpannerinOP}
Given an outerplanar graph $G$ on $n$ vertices and a positive integer $t$,
 a tree $t$-spanner of $G$, if one exists, can be found in $O(n)$ time.
\end{theorem}

From Theorem \ref{treeSpannerinOP},  by using binary search on the value of $t$, we have the following corollary.
 
 \begin{corollary}
\label{Corollary_MSST}
 Given an outerplanar graph $G$ on $n$ vertices,  a minimum stretch spanning tree of $G$ can be found in $O(n \log n)$ time.
 \end{corollary}
 
\noindent
\textbf{Tree partition problems:} In our work, we introduce the \textsc{tree S-partition} problem, a variant of the \textsc{bounded component spanning forest} (see page number 208 in  \cite{GareyJohnsonBook}),  and relate it to tree spanners in outerplanar graphs. The \textsc{bounded component spanning forest} is known to be NP-complete \cite{GareyJohnsonBook}.  To the best of our knowledge, there is no literature on \textsc{tree S-partition} which we have defined as follows.\\

\noindent
\begin{tabular}{|p{16.3cm}|}
 \hline
\textsc{Tree S-partition} \\

\textbf{Instance: }  A tree $T$, a weight function $w: V(T) \to \mathbb{N}$, a set $S \subseteq V(T)$ of special vertices, and $t\in \mathbb{N}$.  \\
\textbf{Question: }  Does there exist a partition of $V(T)$ into sets $V_{1}, \ldots, V_{|S|}$ such that 
for $1\leq i \leq |S|$, \\$T[V_{i}]$ is connected, $|V_{i} \cap S|=1$ and ${\displaystyle \sum_{v \in V_i} w(v)\leq t}$ ?\\
\hline
\end{tabular}\\

Let $\pi = \{V_1, \ldots, V_{|S|}\}$ be a partition of $V(T)$. For each $1 \leq i \leq |S|$, if $T[V_i]$ is a tree and $V_i$ has exactly one vertex from $S$, then $\pi$ is referred to as a \emph{tree S-partition} of $T$.
The \textsc{supply-demand tree partition} problem is defined as follows.\\

\noindent
\begin{tabular}{|p{16.3cm}|}
  \hline
{\textsc{Supply-demand tree partition} \cite{PartitioningTreesSupplyDemand} }\\

\textbf{Instance: } A Tree $T$ such that $V(T) = V_{s} \uplus V_{d}$, a supply function $s: V_{s} \to \mathbb{R}^{+}$, and a demand function $d: V_{d} \to \mathbb{R}^{+}$.\\
\textbf{Question: } 
  Does there exist a partition of $V(T)$ into sets $V_{1}, \ldots, V_{k}$, where $ k = |V_{s}|$,  such that  \\for $1\leq i \leq k$, $T[V_{i}]$ is connected, 
$V_{i}$ contains exactly one vertex $u \in V_{s}$  
and ${\displaystyle \sum_{v \in V_i \setminus V_s} d(v)\leq s(u)}$ ?\\
\hline
\end{tabular}\\

In \textsc{supply-demand tree partition}, each element $u$ in $V_s$ is referred to as supply vertex and $s(u)$ denotes the supply value of $u$. Similarly, each element $u$ in $V_d$ is referred to as demand vertex and $d(u)$ denotes the demand value of $u$.
Let $\pi = \{V_1, \ldots, V_k\}$ be a partition of $V(T)$. For each $1 \leq i \leq k$, if $V_i$ satisfies all the three constraints described, then $\pi$ is referred to as a \emph{supply-demand tree partition} of $T$.

%

\section{Preliminaries}
Throughout this paper, we consider tree spanners in unweighted outerplanar graphs.

\noindent
\textbf{Graph theoretic preliminaries:} In this paper, we consider only simple, finite, connected and undirected  graphs. 
We refer to \cite{dbwestBook} for standard graph theoretic terminologies.
Let $G = (V(G),E(G))$ be a graph on vertex set $V(G)$ and edge set $E(G)$.
 For a set $S \subseteq V(G), G[S]$ denotes the graph induced on the set $S$ and $G - S$ denotes $G[V(G)- S]$.
For a set $V$, the sets $V_1, \ldots ,V_r$ form a partition if and only if $\bigcup_{i=1}^{r}V_{i} = V$ and $\forall i,j$ $1\leq i < j \leq r, V_i \cap V_j = \emptyset$. Each such $V_{i}$ is referred to as a \emph{part} of the partition. The number of edges in a shortest path from a vertex $u$ to a vertex $v$ in $G$ is called the \emph{distance} between $u$ and $v$, and is denoted by $d_G(u,v)$.

\noindent
\textbf{Planar graph preliminaries:}
A graph is said to be \emph{planar} if it can be drawn in the plane so that its edges intersect only at their ends; otherwise it is \emph{nonplanar}. A planar graph is said to be \emph{outerplanar} if it can be embedded on the plane such that all vertices lie on the boundary of its exterior region. Such an embedding is called an \emph{outerplanar embedding}.
Every 2-connected outerplanar graph has a unique outerplane embedding \cite{Syslo82UniueEmbedding}. 
In this paper, a planar graph and an outerplanar graph are associated with  a fixed planar and outerplanar embedding, respectively.
The regions defined by the planar embedding of a planar graph $G$ are \emph{faces} of $G$ and the set of all faces of $G$ is denoted by $Faces(G)$. For each face $f \in Faces(G)$, $V(f)$ and $E(f)$ denote the sets of vertices and edges of $f$, respectively.  Two faces $f, f' \in Faces(G)$ are \emph{adjacent} if $E(f) \cap E(f') \neq \emptyset$.
In every planar embedding, there is a unique unbounded face called \emph{exterior} face and it is denoted by $f_{ext}$.
An edge $e$ is \emph{external} if $e \in E(f_{ext})$, otherwise it is \emph{internal}.
All the bounded faces of $G$ are \emph{interior}. Interior faces that are adjacent to the exterior face are called $E\emph{-}faces$ and other interior faces are called $I\emph{-}faces$. The set of interior faces, $I\emph{-}faces$ and $E\emph{-}faces$ of $G$ are denoted by  $Int\emph{-}faces(G)$,$I\emph{-}faces(G)$, and $E\emph{-}faces(G)$, respectively. Also, $Faces(G) = \{f_{ext}\}$ $\cup$ $Int\emph{-}faces(G)$ and $Int\emph{-}faces(G)= E\emph{-}faces(G)$ $\cup$ $I\emph{-}faces(G)$. An example is shown in Figure \ref{PlanarPrelims} to illustrate these concepts.

\begin{figure*}[htp!]
\centering
\includegraphics[scale=0.4]{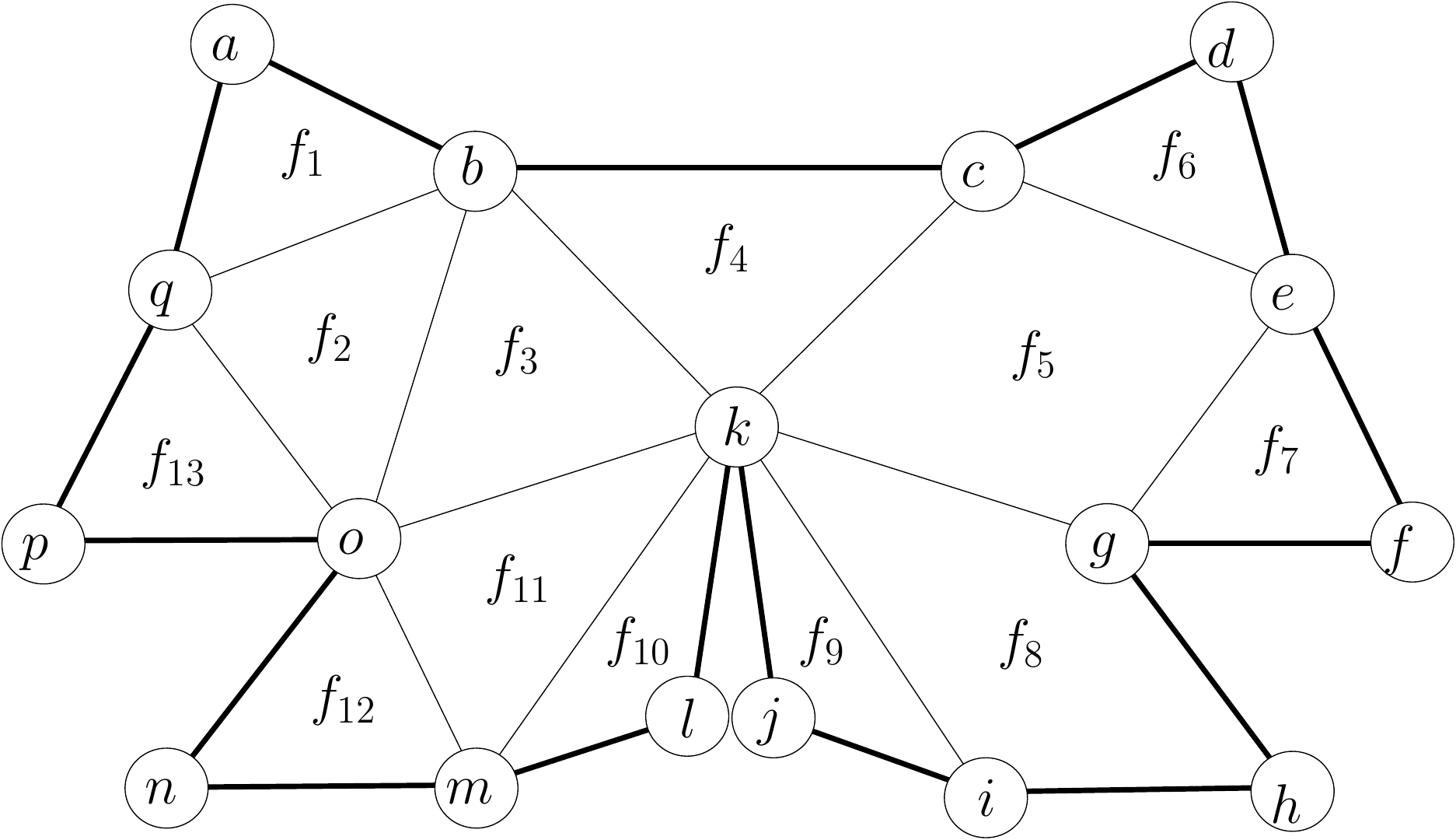}
\caption{For the outerplanar graph $G$ shown, thick edges are external and other edges are internal.  $E\emph{-}faces(G) = \{f_1,f_4,f_6,f_7,f_8,f_9,f_{10},f_{12},f_{13} \}$, $I\emph{-}faces(G) = \{ f_2,f_3,f_5,f_{11} \}$. }
\label{PlanarPrelims}
\end{figure*}

\begin{lemma}
(Theorem 1.1 in \cite{Cai95treespanners})
\label{treeSpannersProp}
A spanning tree $T$ of an arbitrary graph $G$ is a tree $t$-spanner if and only if for every non-tree edge $(x,y)$ of $T$, $d_{T}(x,y) \leq t$.
\end{lemma}

 \noindent
Let $T$ be a spanning tree of a graph $G$. For any edge $(u,v) \in E(G) \setminus E(T)$, let $e=(u,v)$, the stretch of $e$ is defined as the distance between $u$ and $v$ in $T$. The stretch of $e$ is denoted by $\str(e)$. If $e \in E(T)$ then $\str(e) = 1$, otherwise $\str(e) = d_T(u,v)$. 
By Lemma \ref{treeSpannersProp}, the stretch of  $T$ is $\max \{ \str(e) \mid e \in E(G) \setminus E(T) \}$.
For vertices $u,v \in V(T)$, $P_{T}(u,v)$ denotes the path between $u$ and $v$ in $T$.
An edge $e \in E(G) \setminus E(T)$ is a \emph{non-tree} edge of $T$.
For a non-tree edge $(u,v)$  of $T$, the cycle formed by $(u,v)$ and $P_{T}(u,v)$ is referred to as a \emph{fundamental} cycle. 
A non-tree edge $e$ of $T$ is said to be \emph{external} if $e$ is external in $G$.
A fundamental cycle is said to be \emph{external} if the associated non-tree edge is external. 

It is easy to see that the vertex connectivity of an outerplanar graph is at most two, because the end vertices of every internal edge form a vertex separator. If an outerplanar graph $G$ has cut vertices, a minimum stretch spanning tree of $G$ can be obtained by performing a union on the edges of minimum stretch spanning trees of the maximal 2-vertex-connected components of $G$. So without loss of generality, we consider 2-vertex-connected outerplanar graphs.

 In the rest of the paper, $G$ denotes a 2-vertex-connected outerplanar graph. We refer to the internal and external edges with respect to the graph $G$.  The same holds for $E\emph{-}faces$, $I\emph{-}faces$ and $Int\emph{-}faces$.

\section{Canonical Tree Spanners in Outerplanar Graphs}
\label{SectionProperties}
In this section, we identify a canonical structure that is respected by some tree $t$-spanner: in
each $E\emph{-}face$, exactly one external edge is a non-tree edge.  This is formally proved in 
Theorem \ref{NumPartsEface}  which is the main result in this section.
\begin{lemma}
\label{LemmaPartitionNonTreeEdges}
Let $T$ be an arbitrary spanning tree of $G$.
Let $\{C_1, \ldots, C_k\}$ be the set of external fundamental cycles of $T$. 
For each $i$, let $G_i = G[V(C_i)]$. \\
a. $Int\emph{-}faces(G) = Int\emph{-}faces(G_1) \uplus \ldots \uplus Int\emph{-}faces(G_k)$.\\
b. The set of non-tree edges of $T$ in $G$ is a disjoint union of the set of non-tree edges of $T$ in $G_1, \ldots, G_k$.
\end{lemma}
\begin{proof}
Let $r$ be the number of edges in $T$ that are internal edges in $G$.
The proof is by induction on $r$.
Suppose $r=0$, then all the edges in $T$ are external edges in $G$.
Then $T$ has exactly one external fundamental cycle, which is also an exterior face of $G$.  Consequently, the two claims in this lemma hold.
Consider the case when $r\geq 1$.
Then there exists an edge $(u,v) \in  E(T)$, such that  $(u,v)$ is internal in $G$.
We decompose $G$ into graphs $H_1$ and $H_2$ such that
$V(G) = V(H_1) \cup V(H_2)$, $E(G) = E(H_1) \cup E(H_2)$, $V(H_1) \cap V(H_2) = \{u,v\}$ and $E(H_1) \cap E(H_2) = \{ (u,v) \}$.
This decomposition is possible as the end vertices of any internal edge in an outerplanar graph disconnects the graph exactly into two components.
Observe that $Int\emph{-}faces(G) = Int\emph{-}faces(H_1) \uplus Int\emph{-}faces(H_2)$.

For $i \in \{1,2\}$, let $T_i = T[V(H_i)]$. For some fixed $i \in \{1,2\}$, let $x,y \in V(T_i)$ and $(u,v) \neq (x,y)$, we now show that there is a path between $x$ and $y$ in $T_i$.
Let $P(x,y)$ be the path between $x$ and $y$ in $T$.
If all the vertices in $P(x,y)$ are in $T_i$, then we are done.
Consider the case that there is a vertex $z$  in $P(x,y)$, such that $z$ is in  $T_j$ but not in $T_i$, where $i\neq j$.
Since the vertices $x,y \in V(T_i)$ and $z \in V(T_j)$, and we know that $\{u,v\}$ is a minimal vertex separator of $G$, it follows that there is a subpath $P(u,v) \subseteq P(x,y)$ in $T$, such that $P(u,v)$ has  the vertex $z$.
Therefore, $P(u,v)$ has at least two edges in $T$. As $(u,v)$ is in $T$, $P(u,v) + (u,v)$ forms a cycle in $T$. This contradicts that $T$ is an acyclic graph. As a result, for $i \in \{1,2\}$, $T_i$ is a spanning tree of $G_i$.
Since $E(H_1) \cap E(H_2)$ is $(u,v)$, the set of non-tree edges of $T$ is a disjoint union of the set of non-tree edges of $T_1$ in $H_1$ and the set of non-tree edges of $T_2$ in $H_2$.
 It implies that the external fundamental cycles of $T$ is a disjoint union of  external fundamental cycles of $T_1$ and external fundamental cycles of $T_2$.
Without loss of generality, assume that $C_1, \ldots, C_l$ be the external fundamental cycles in $T_1$ and $C_{l+1}, \ldots, C_k$ be the external fundamental cycles in $T_2$.

By the induction hypothesis, $Int\emph{-}faces(H_1) = Int\emph{-}faces(G_1) \uplus \ldots \uplus Int\emph{-}faces(G_l)$ and $Int\emph{-}faces(H_2) = Int\emph{-}faces(G_{l+1}) \uplus \ldots \uplus Int\emph{-}faces(G_k)$.
As a result, $Int\emph{-}faces(G) = Int\emph{-}faces(G_1) \uplus \ldots \uplus Int\emph{-}faces(G_k)$.
Also by the induction hypothesis, the set of non-tree edges of $T_1$ in $H_1$ is a disjoint union of the set of non-tree edges of $T_1$ in $G_1, \ldots, G_{l}$ and the set of non-tree edges of $T_2$ in $H_2$ is a disjoint union of the set of non-tree edges of $T_2$ in $G_{l+1}, \ldots, G_{k}$.
Consequently, the set of non-tree edges of $T$ in $G$ is a disjoint union of the set of non-tree edges of $T$ in $G_1, \ldots, G_k$.
\qed
\end{proof}

\noindent
For an illustration of Lemma \ref{LemmaPartitionNonTreeEdges}, consider the outerplanar graph $G$ shown in Fig \ref{FigPartition}.
Let $G_1 = G[a,b,o,p,q]$, $G_2 = G[b,c,k,o]$, $G_3 = G[c,d,e,f,g,k]$, $G_4=G[g,h,i,j,k]$, $G_5=G[k,l,m,n,o]$.
Observe that $Int\emph{-}faces(G_1)$ $ = \{f_1, f_2, f_{13} \}$, $Int\emph{-}faces(G_2) = \{f_3, f_4\}$, $Int\emph{-}faces(G_3) = \{ f_5, f_6, f_7\}$, $Int\emph{-}faces(G_4) = \{ f_8, f_9\}$ and $Int\emph{-}faces(G_5) = \{ f_{10}, $ $f_{11}, f_{12}\}$.
Further, $Int\emph{-}faces(G)$ is  a partition of $Int\emph{-}faces(G_1), \ldots, Int\emph{-}faces(G_5)$.

\begin{figure*}[htp!]
\centering
\includegraphics[scale=0.4]{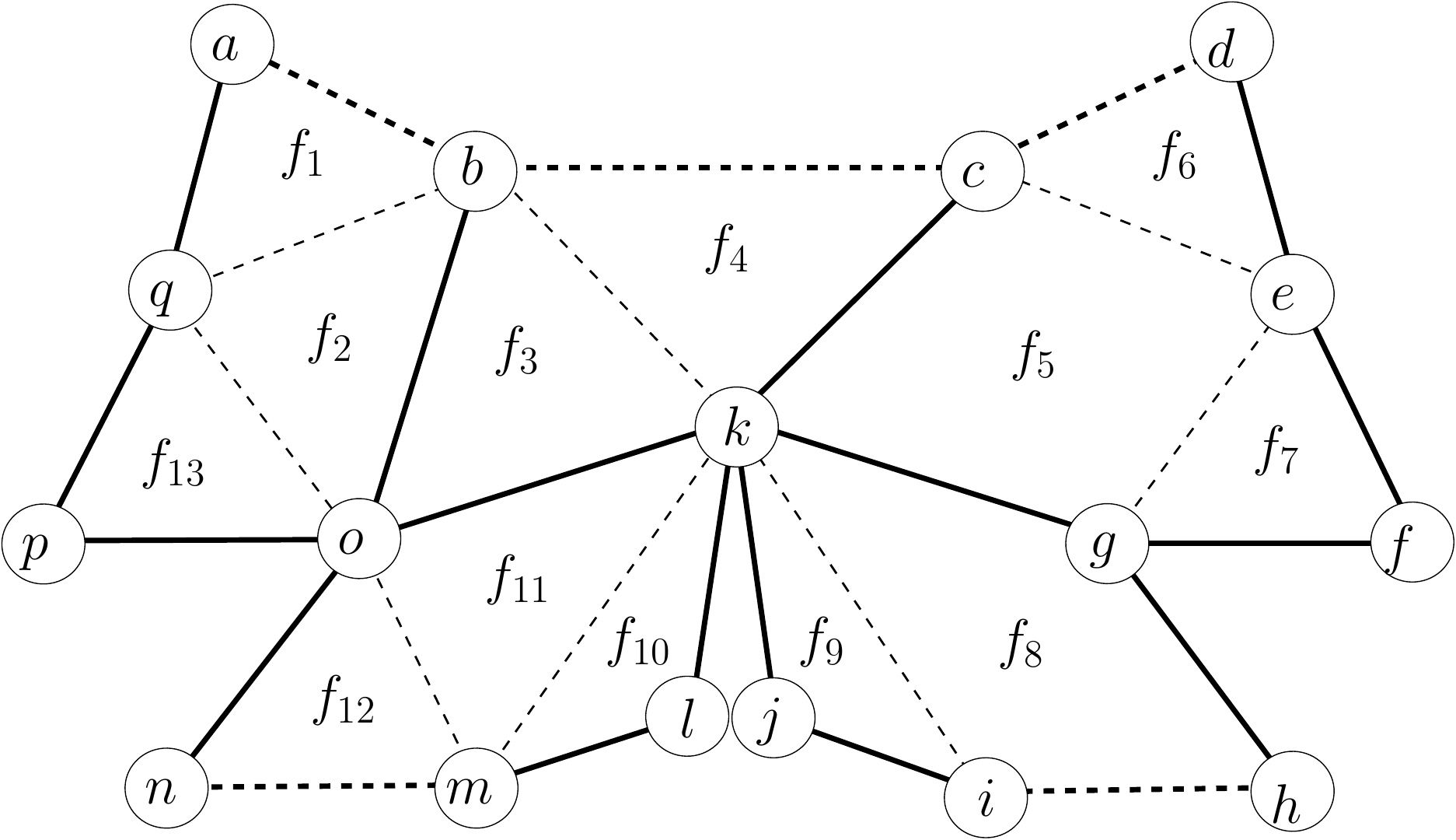}
\caption{For the outerplanar graph $G$ shown, an arbitrary spanning tree $T$ of $G$ is presented in thick edges. Dotted edges are non-tree edges of $T$. Dotted edges shown in thick are non-tree edges of $T$ that are external in $G$.}
\label{FigPartition}
\end{figure*}


We strengthen the result of Lemma \ref{treeSpannersProp} in the lemma below, which establishes a necessary and sufficient condition for the existence of a tree $t$-spanner in outerplanar graphs. 

\begin{lemma}
\label{Lemma_StretchAtExternalEdge}
Let $T$ be a spanning tree of $G$. $T$ is a tree $t$-spanner if and only if  for every external non-tree edge $(u,v)$, $d_{T}(u,v) \leq t$.
\end{lemma}
\begin{proof}
($\Rightarrow$) This follows from Lemma \ref{treeSpannersProp}.

($\Leftarrow$)  Let $C_1, \ldots, C_k$ be the external fundamental cycles of $T$.
For each $i$, let $C_i$ be an external fundamental cycle formed by a tree path $P_i$ and an external non-tree edge $e_i$, and let $G_i = G[V(C_i)]$.
By Lemma \ref{LemmaPartitionNonTreeEdges}~b, for each non-tree edge $e \in E(G) \setminus E(T)$, there is a unique $i$ such that $e \in E(G_i)$.
For each non-tree edge $e$ in $G_i$, if $e$ is an internal edge in $G$, then $\str(e) < \str(e_i)$ as $P_i$ is a Hamilton path in $G_i$. 
Thereby, for every non-tree edge $e \in E(G) \setminus E(T)$, such that $e$ is internal in $G$, $\str(e) < \str(e_i)$ for some $1 \leq i \leq k$.
Recall that the stretch of a spanning tree is the maximum stretch over all of its non-tree edges.
Therefore, the stretch of $T$ is maximum stretch over all of the non-tree edges of $T$ that are external in $G$.
\qed
\end{proof}

\noindent
For a cycle $C$ in $G$, the \emph{enclosed region} of $C$  is defined as the set of interior faces in $G[V(C)]$ and is denoted by $Enc(C)$. 
For instance, for the cycle  $C$ on the vertices in order $a,b,o,p,q$ shown in  Figure \ref{PlanarPrelims}, $Enc(C)$ is $\{f_1,f_2,f_{13} \}$.

\begin{theorem}
\label{NumPartsEface}
 Let $G$ admit a tree $t$-spanner. Then there exists a tree $t$-spanner $T$ of $G$ that satisfies the following properties:\\
$(P1)$.For every  $E\emph{-}face$ $f$ of $G$, there exists exactly one external edge $e \in E(f)$ such that $e \notin E(T)$.\\
$(P2)$.For every external edge $(u,v) \notin E(T)$, there exists exactly one $E\emph{-}face$ of $G$ in $G[V(C)]$, where $C= P_{T}(u,v) + (u,v)$.
\end{theorem}
\begin{proof}
Let $T'$ be an arbitrary tree $t$-spanner of $G$. We obtain a tree $t$-spanner $T$ that satisfies $P1$ by transforming $T'$.
 If $|Int\emph{-}faces(G)|=1$, then $T'$ satisfies $P1$. So let us consider the case where $|Int\emph{-}faces(G)| \geq 2$.
Let $f \in E\emph{-}faces(G)$. If two or more external edges of $f$ are not present in $T'$, then it follows that $T'$ is not connected which leads to a contradiction that $T'$ is indeed connected. Suppose there exists a face $f \in E\emph{-}faces(G)$, such that all the external edges in $E(f)$  are present in $E(T')$. 
Let $e \in E(f)$ be an external edge.
Let $(T_1,T_2)$ be the cut in $G$ obtained by the removal of the edge $e$ from $T'$. Since $G$ is outerplanar and 2-vertex-connected, the number of edges in the face $f$ that crosses the cut $(T_1,T_2)$ is exactly two.  One such edge is $e$ and let $e'$ be the other edge. 
Clearly, $T' + e'- e$ is a spanning tree of $G$ and we show that its stretch is not more than that of $T'$. Let $C$ be an external fundamental cycle of $T'$ such that the face $f$ is contained in $Enc(C)$. 
The existence of $C$ is clear from Lemma \ref{LemmaPartitionNonTreeEdges}~a.
Let the cycle $C$  be formed with an external non-tree edge $(x,y)$ and the path $P_{T'}(x,y)$. 
Adding $e'$ to $T'$ and removing $e$ from $T'$ is equivalent to splitting the external fundamental cycle $C$ into two external fundamental cycles $C_1$ and $C_2$, so that $Enc(C_1)$ is $\{f\}$ and $Enc(C_2)$ is $Enc(C) \setminus \{f\}$. An example of this process is shown in Figure \ref{canonicalFigure}. 
It follows that the set of external fundamental cycles of $T'$ other than $C$, remain same in $T' + e'- e$.
Note that external non-tree edges associated with external fundamental cycles $C_1$ and $C_2$ are $e$ and $(x,y)$, respectively.
The stretch of $e$ and $(x,y)$ is not more than $t$, as $G[V(C)]$ is outerplanar and the interior faces in $Enc(C)$ are distributed to $Enc(C_1)$ and $Enc(C_2)$.
 From Lemma \ref{Lemma_StretchAtExternalEdge}, stretch of $T'+e'-e$ is not more than the stretch of $T'$.
By repeated application of the above process for at most $|E\emph{-}faces(G)|$ times, we obtain a tree $t$-spanner $T$ that satisfies $P1$.

We now prove that $T$ satisfies $P2$.
Let $(u,v) \in E(G)\setminus E(T)$ be an external edge and $e=(u,v)$. Let $C_{e} = P_{T}(u,v) + (u,v)$ and $H=G[V(C_{e})]$.
 If $H$ contains no $E\emph{-}face$ of $G$, then it implies that $e$ is an internal edge, a contradiction to the premise that $e$ is an external edge. Assume that $H$ contains at least two $E\emph{-}faces$ $f$ and $f'$ of $G$. Without loss of generality, let $E(f)$ contain $e$. Then  all the external edges of $f'$ are in $T$, a contradiction to the fact that $T'$ satisfies $P1$. Hence the lemma is proved.
 \qed
\end{proof}

\begin{figure*}[htp!]
\centering
\includegraphics[scale=0.7]{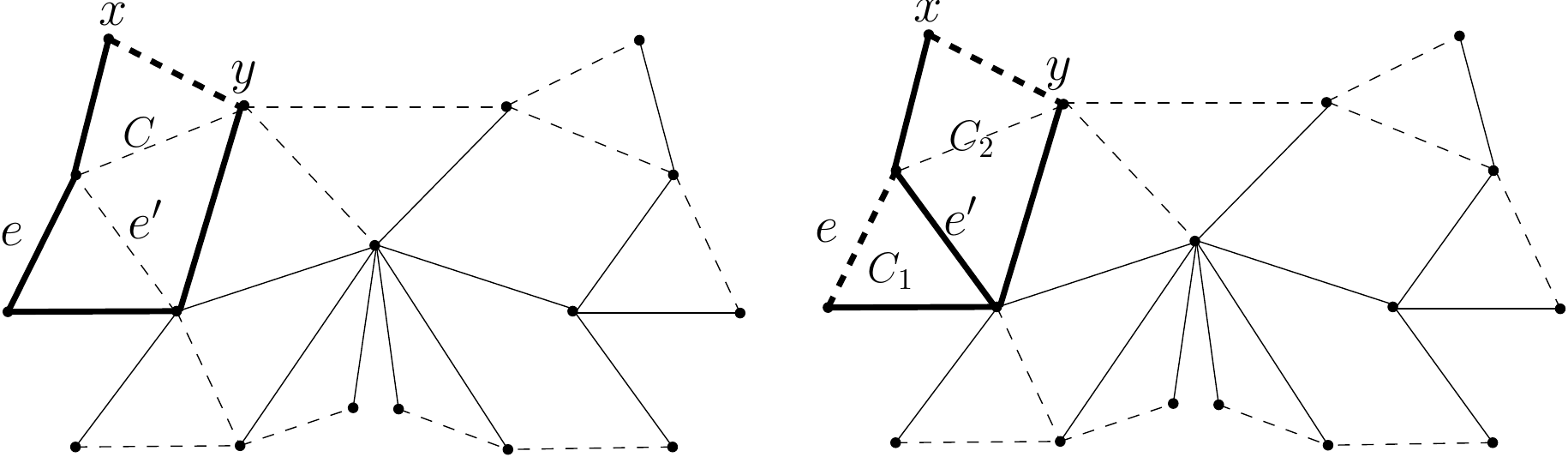}
\caption{(Left) An outerplanar graph $G$ and a non-canonical tree 4-spanner of $G$. (Right) An outerplanar graph $G$ and a canonical tree 4-spanner of $G$. (Dotted edges represent the non-tree edges and the rest of the edges are tree edges)}
\label{canonicalFigure}
\end{figure*}

\noindent
A tree $t$-spanner of $G$ that satisfies the properties $P1$ and $P2$ of Theorem \ref{NumPartsEface} is called a \emph{canonical} tree $t$-spanner. From Theorem \ref{NumPartsEface}, it is clear that the number of external edges missing in a canonical tree $t$-spanner of $G$ is equal to the number of $E\emph{-}faces$ of $G$.

\section{Tree $t$-spanner via Tree S-partition}
\label{treeSpannersSpartitionSection}
In this section, we present a linear-time reduction from \textsc{Tree $t$-spanner} in an outerplanar graph $G$  to \textsc{tree S-partition} in $\tilde{T}$, where  $\tilde{T}$ is a weak-dual of $G$. Further, we present a linear-time algorithm to construct a tree $t$-spanner of $G$ from a tree S-partition of $\tilde{T}$.
To establish the connection between these two problems, we need the following lemmas.
\begin{lemma} (Lemma 2.4 in \cite{coloringOuterplanar})
\label{outerplanarWeakDualTree}
 An outerplanar graph $G$ is $2$-connected if and only if the weak dual of $G$ is a tree.
\end{lemma}

\begin{lemma}
 (Lemmas 5 and 6 in \cite{MCBOuterPlanar98})
 \label{opH2_connected}
An outerplanar graph is Hamiltonian if and only if it is 2-connected. Moreover, every 2-connected outerplanar graph has a unique Hamiltonian cycle.
\end{lemma}

\subsection{Linear-time Reduction}
\noindent \textbf{Reduction.}
Reduction from Tree $t$-spanner in outerplanar graphs to Tree S-partition is as follows.
\label{OuterPlanarreductionSubSection}
 For a given instance $\langle G, t \rangle$ of \textsc{tree $t$-spanner}, an instance $\langle \tilde{T},w,S,t-1 \rangle$ of  \textsc{tree S-partition}   is constructed as follows: $V(\tilde{T}) = \{v_f \mid f \in Int\emph{-}faces(G)\}$, $E(\tilde{T}) = \{ (v_f, v_g) \mid f,g \in Int\emph{-}faces(G), |E(f)\cap E(g)| = 1 \}$. The weight function $w: V(\tilde{T}) \to \mathbb{N}$ is defined as  $ w(v_f) = |E(f)|-2$, for each $v_f \in V(\tilde{T})$. The set of special vertices $S$ is $\{ v_{f} \mid f \in $ $E\emph{-}faces(G) \}$. The graph $\tilde{T}$ is the \emph{weak dual} of $G$ and it is a tree due to Lemma \ref{outerplanarWeakDualTree}. 

To show that the described transformation takes linear time, we first describe an approach to find interior faces in outerplanar graphs with the help of a minimum cycle basis.
We refer the reader to \cite{MCBOuterPlanar98} for details on cycle basis.
A minimal set $\mathcal{B}$ of cycles is a \emph{cycle basis} of $G$, if every cycle in $G$ can be expressed as exclusive-or sum of a subset of cycles in $\mathcal{B}$. 
The length of a cycle basis is sum-total length of cycles in the cycle basis.
A cycle basis with minimum length is a \emph{minimum cycle basis}.  Outerplanar graphs have a unique minimum cycle basis and this can be obtained in linear time.

\begin{lemma}
\label{LemmaMCBop}
\cite{MCBOuterPlanar98}
 For an unweighted outerplanar graph on $m$ edges and $n$ vertices, there is a unique minimum cycle basis of length $2m-n$, which can be obtained in linear time.
\end{lemma}

\begin{lemma}
\label{Lemma_faceOuterPlanar}
The set of interior faces of a 2-connected outerplanar graph can be obtained in linear time.
\end{lemma}
\begin{proof}
Let $G$ be an outerplanar graph on $m$ edges and $n$ vertices.
We  observe that the set of interior faces of $G$ is a cycle basis of length $2m-n$. Therefore, by Lemma \ref{LemmaMCBop}, we conclude that obtaining the set of interior faces of an outerplanar graph takes linear time.
\qed
\end{proof}

\begin{lemma}
\label{Lemma_runningTime}
Given an instance $\langle G, t\rangle$ of \textsc{tree $t$-spanner}, the transformed instance $\langle \tilde{T}, w, S, t-1 \rangle$ of \textsc{tree S-partition} can be obtained in linear time.
\end{lemma}
\begin{proof}
We set up the necessary data-structure for the reduction. 
For an edge $e$ in $G$, let $Faces[e]$ denote the set of at most two interior faces of $G$ that contain $e$.
For an interior face $f$ in $G$, let $Num[f]$ denote the number of edges in $f$. 
For each interior face $f$ in $G$ and for each edge $e \in E(f)$, insert $f$ in the set $Faces[e]$.
Further, assign $Num[f] = |E(f)|$. 
We now construct $\tilde{T},w$ and $S$. For each interior face $f$ in $G$, add a vertex $v_{f}$ in $\tilde{T}$. For each edge $e\in E(G)$ such that $|Faces[e]|$ is two, let $Faces[e] = \{f,g\}$, add an edge $(v_f,v_g)$ in $\tilde{T}$. For each $f \in Int\emph{-}faces(G)$, assign $w(v_f)= Num[f]-2$. A vertex $v_f \in V(\tilde{T})$ is considered to be in $S$ if the degree of $v_f$ is less than $Num[f]$. 

The set up of the data structures $Faces[~]$ and $Num[~]$ takes linear time as every edge appears in at most two interior faces and  $|E(G)| \leq 2|V(G)|-3$ in outerplanar graphs.
Therefore, the construction of $\langle \tilde{T}, w, S, t-1 \rangle$ takes linear time.
\qed
\end{proof}

\subsection{Correctness of Reduction}

We use the following lemmas to prove formally, in Theorem \ref{outerplanarReductionTheorem}, that the described reduction is correct .

\begin{lemma}
\label{subgraphOuterPlanar}
Let $\tilde{T}$ be the weak dual of $G$. 
Let $F' \subset Int\emph{-}faces(G)$, $X = \bigcup_{f \in F'} V(f)$ and $V' = \{ v_{f} ~|~ f \in F' \}$.
The outerplanar graph $G[X]$ is $2$-connected if and only if $\tilde{T}[V']$ is a tree.
\end{lemma}
\begin{proof}
We first prove the claim that $\tilde{T}[V']$ is the weak dual of $G[X]$. Clearly, the vertex set $V'$ is same as the vertex set of the weak dual of $G[X]$. For any two interior faces $f$ and $g$ in $G[X]$, $f$ and $g$ are adjacent in $G$ if and only if they are adjacent in $G[X]$. Hence, $(v_f,v_g)$ is an edge in $\tilde{T}[V']$  if and only if it is an edge in the weak dual of $G[X]$. Hence the claim holds. 
By applying Lemma \ref{outerplanarWeakDualTree}, from the above claim, $G[X]$ is $2$-connected if and only if  $\tilde{T}[V']$ is a tree.
\qed
\end{proof}

The following lemma gives the length of the Hamiltonian cycle in terms of length of interior faces.  

\begin{lemma}
\label{longestCycle}
 Let $l_1, \ldots, l_r$ be the lengths of all the interior faces in $G$. If $C$ is the Hamiltonian cycle in $G$, then $\displaystyle{|C| = (\Sigma_{i=1}^{r}l_{i})-2(r-1)}$.

\end{lemma}
\begin{proof}
Let $l_{ext} = |V(f_{ext})|$. Since $2|E(G)|= l_1 + \ldots + l_r + l_{ext}$, by Euler's planarity formula, we obtain
 $|C| = |V(G)|$ $= |E(G)| -r + 1$ $ = 2|E(G)| - 2r + 2 -|V(G)| $ $= l_1 + \ldots + l_r - 2(r-1) + l_{ext} -|V(G)|$.
Since $G$ is  outerplanar, $l_{ext} = |V(G)|$, thereby $\displaystyle{|C| = (\Sigma_{i=1}^{r}l_{i})-2(r-1)}$.
\qed
\end{proof}

For a subset $X \subseteq V(\tilde{T})$, we define $cost(X) = {\displaystyle \Sigma_{v \in X}w(v)}$.
For a partition $\pi = \{V_1, \ldots, V_k \}$ of $V(\tilde{T})$, $cost(\pi)$ is defined as $ \max \{ cost(V_{i}) \mid V_{i} \in \pi \}$.

\begin{lemma}
 \label{longestCycleEquation}
Let $C$ be a cycle in $G$ and  $V' = \{ v_{f} \in V(\tilde{T}) \mid f \in Enc(C) \}$.
Then $cost(V') = |C|-2$.
 \end{lemma}
\begin{proof}
Let $Enc(C) = \{f_1, \ldots, f_r \}$ and $l_i = |E(f_i)|$. By premise, $|V'| = r$. From the transformation, we have
 $cost(V') = {\Sigma_{v \in V'}w(v)} = \Sigma_{i=1}^r (l_{i}-2) = \Sigma_{i=1}^r l_{i}-2(r-1) -2$. By Lemma  \ref{longestCycle}, we have $\Sigma_{i=1}^r l_{i}-2(r-1) -2 = |C| - 2$. Hence, $cost(V') = |C|-2$.
\qed
\end{proof}


\begin{theorem}
\label{outerplanarReductionTheorem}
 $G$ admits a tree $t$-spanner if and only if $\tilde{T}$ has a tree S-partition of cost at most $t-1$. 
  \end{theorem}
\begin{proof} \textbf{Necessity.}
Let $T$ be a canonical tree $t$-spanner of $G$.
From $T$, we construct a tree S-partition $\pi$ of $\tilde{T}$ such that $cost(\pi) \leq t-1$.
 For each external edge $(u,v)$ such that $(u,v) \in E(G) \setminus E(T)$, let $e=(u,v)$,
 we construct a part $V_{e} \in \pi$ with the desired properties as follows.
Consider the cycle $C = P_{T}(u,v) + (u,v)$ in $G$. Clearly, $|C| \leq t+1$ because $|P_{T}(u,v)| \leq t$.
Note that $G_{e} = G[V(C)]$ is a 2-connected outerplanar graph.
Let $V_{e} = \{ v_{f} ~|~ f \in Int\emph{-}faces(G_{e})\}$.
Now $V_e$ is a part in $\pi$. By Theorem \ref{NumPartsEface}, there is exactly one face $f$ in $G_{e}$ such that $f$ is an $E\emph{-}face$ in $G$. By the transformation, $v_{f} \in V_e$ and since $v_{f} \in S$, we have $|V_{e} \cap S| = 1$.
Moreover, it is easy to see that $\tilde{T}[V_{e}]$ is connected from Lemma \ref{subgraphOuterPlanar}.
From Lemma \ref{longestCycleEquation}, $cost(V_{e}) = |C|-2$. So, $cost(V_{e}) \leq t-1$.
From Theorem \ref{NumPartsEface}, the number of parts in $\pi$ is $|S|$. 
By Lemma \ref{LemmaPartitionNonTreeEdges}~a, $\pi$ is a partition of $V(\tilde{T})$.
Hence $\tilde{T}$ has a tree S-partition of cost at most $t-1$.

\textbf{Sufficiency.}
Let $\pi = \{ V_{1}, V_{2}, \ldots, V_{q}\}$ be a tree S-partition of $\tilde{T}$ such that $cost(\pi) \leq t-1$, where $|S|=q$.
From $\pi$, we construct a tree $t$-spanner of $G$.
For each $1 \leq i \leq q$, let $G_i$ be a 2-connected outerplanar graph, where $V(G_i) = \bigcup_{ v_{f} \in V_{i} } V(f)$ and $E(G_i) = \bigcup_{ v_{f} \in V_{i} } E(f)$. As $G_i$ is a subgraph of $G$, clearly $G_i$ is outerplanar. Further by Lemma \ref{subgraphOuterPlanar}, $G_{i}$ is $2$-connected, because $\tilde{T}[V_{i}]$ is a tree. Thus $G_i$ is a 2-connected outerplanar graph. 
From Lemma \ref{opH2_connected}, $G_{i}$ has a unique Hamilton cycle, say $C_i$.
Let $X_i$ be the set of internal edges in $G_i$.
Because the set of external edges in a 2-connected outerplanar graph forms a Hamilton cycle, $E(G_i)$ is a disjoint union of $E(C_i)$ and $X_i$. 
Since $V_1, \ldots, V_q$ is a partition of $V(\tilde{T})$, observe that $Enc(C_1), \ldots, Enc(C_q)$ is a partition of $Int\emph{-}faces(G)$.
We obtain the 2-connected outerplanar graph $G'$ by removing the edges $X_1 \cup \ldots \cup X_q$ from $G$.
Since $G'$ does not contain internal edges from any $G_i$, we have $I\emph{-}faces(G') = \emptyset$ and $E\emph{-}faces(G')=\{C_1, \ldots, C_q\}$. 
Let $e_i \in E(C_i)$ be an external edge in $G'$.
We remove the set $\{e_1, \ldots, e_q \}$ of edges from $G'$ and obtain $T$.
Since we have deleted exactly one external edge from each $E\emph{-}face$ of $G'$ and there is no $I\emph{-}face$ in $G'$, $T$ is a spanning tree. It remains to show that $T$ is a tree $t$-spanner of $G$.
Note that $cost(V_{i}) \leq  t-1$, because $V_{i} \in \pi$.
From Lemma \ref{longestCycleEquation}, $|C_{i}|=cost(V_{i})+2$, so $|C_{i}| \leq t+1$.
It follows that for each $1 \leq i \leq q$, stretch of $e_i$ in $T$ is at most  $t$. 
As we can observe that $e_1, \ldots, e_q$ are the only non-tree edges of $T$ that are external in $G$, 
from  Lemma \ref{Lemma_StretchAtExternalEdge}, $T$ is a tree $t$-spanner of $G$. 
\qed
\end{proof}


\subsection{Tree $t$-spanner from Tree S-partition in Linear Time}
We now show that the construction in the sufficiency part of Theorem \ref{outerplanarReductionTheorem} can be implemented in
linear time. 

\begin{theorem}
\label{Theorem_ConstructionFromSpartition}
 From a tree S-partition $\pi$ of $\tilde{T}$ with $cost(\pi) \leq t-1$,  a tree $t$-spanner of $G$ can be constructed in $O(|V(G)|)$ time.
\end{theorem}
\begin{proof}
Given a tree S-partition $\pi = \{V_1, \ldots, V_q\}$ of $\tilde{T}$ with $cost(\pi) \leq t-1$,
we present a linear-time algorithm to construct a tree $t$-spanner of $G$. 
For $1 \leq i \leq q$, let $G_i$ be a 2-connected outerplanar graph, where $V(G_i) = \bigcup_{ v_{f} \in V_{i} } V(f)$ and $E(G_i) = \bigcup_{ v_{f} \in V_{i} } E(f)$ and $C_i$ be the Hamilton cycle in $G_i$.  For the vertex sets $V_1, \ldots, V_q$ in $\tilde{T}$, the corresponding parts in $G$ are $G_1, \ldots, G_q$.

For each $e \in E(G)$, we use $count[e]$ to store the number of interior faces in $G$ that contain $e$; $part[e]$ is used store the set of indices $i$ from $\{1,\ldots, q\}$, such that $G_i$ has $e$.
For each $1 \leq i \leq q$, we use $X_i$ to store the set of internal edges in $G_i$.

\textbf{Step 1} Compute  $Int\emph{-}faces(G)$ in linear time (cf. Lemma \ref{Lemma_faceOuterPlanar}). 
 For each $e \in E(G)$, initialize $count[e]=0$ and $part[e] = \emptyset$. For each $1 \leq i \leq q$, initialize $X_i = \emptyset$.

\textbf{Step 2}  For each $1 \leq i \leq q$, for each $v_f \in V_i$, and for each $e \in E(f)$, increment $count[e]$ by one and if $i \notin part[e]$, then insert the index $i$ into $part[e]$. 
For each edge $e \in E(G)$ and if $count[e]=2$  and $|part[e]|=1$, then $X_i \leftarrow X_i \cup \{e\}$, where $i = part[e]$. \\
\texttt{/* $e$ is an interior edge in $G_i$ if and only if $count[e]=2$ and  $part[e]=\{i\}$. \\ Therefore, for $1 \leq i \leq q$, $X_i$ is the set of interior edges in $G_i$ and $E(G_i) = E(C_i) \uplus X_i$. */}

\textbf{Step 3} $G' \leftarrow G -  (X_1 \cup \ldots \cup X_q)$. \\
\texttt{/* $Int\emph{-}faces(G') = E\emph{-}faces(G') = \{C_1, \ldots, C_q\}$ */}.

\textbf{Step 4} Compute $E\emph{-}faces(G') = \{C_1, \ldots, C_q\}$ in linear time (Lemma \ref{Lemma_faceOuterPlanar} is applicable). For each $1 \leq i \leq q$, $Ex(C_i) \leftarrow$ find the set of edges in face $C_i$ that are external in $G'$ in linear time (cf. Lemma \ref{Lemma_findExteriorEdges}). For each $1 \leq i \leq q$, $e_i \leftarrow$ choose an edge from $Ex(C_i)$.

\textbf{Step 5} $T \leftarrow G' - \{e_1, \ldots, e_q\}$. $T$ is a tree $t$-spanner of $G$.

In Step 2, for each $e \in E(G)$, $count[e]$ and $part[e]$ are updated at most two times, because every edge is there in at most two faces. Further, $|E(G)| \leq 2|V(G)| -3$. Thus Step 2 takes $O(|V(G)|)$ time. The rest of the steps clearly take linear time. 
Thus a tree $t$-spanner of $G$ is obtained in $O(|V(G)|)$ time.
\qed
\end{proof}

\begin{lemma}
\label{Lemma_findExteriorEdges}
For each $f \in Int\emph{-}faces(G)$, the set of edges of $f$ that are external in $G$ can be obtained in linear time.
\end{lemma}
\begin{proof}
For each $e \in E(G)$, we use $count[e]$ to count the number of interior faces containing $e$.
For each $e \in E(G)$, initialize $count[e]=0$.
For each $f \in Int\emph{-}faces(G)$ and for each $e \in E(f)$, increment $count[e]$ by one. 
For each $f \in Int\emph{-}faces(G)$, we use $X[f]$ to store the edges of $f$ that are external in $G$.
For each $f \in Int\emph{-}faces(G)$, initialize $X[f] = \emptyset$.
For each $f \in Int\emph{-}faces(G)$ and for each $e \in E(f)$, if $count[e]=1$, then $X[f] \leftarrow X[f] \cup \{e\}$.
For each $f \in Int\emph{-}faces(G)$, the set of edges in $X[f]$ are the edges of $f$ that are external in $G$, because the edges in $X[f]$ appear only in the interior face $f$. 
As each edge in outerplanar graph appears in at most two interior faces, all the steps together take linear time. 
\qed
\end{proof}

We have shown in this section that a tree $t$-spanner for $G$ can be obtained by solving \textsc{tree S-partition} in the weak dual of $G$ in linear time. In the following section we show that \textsc{tree S-partition} can be solved in linear time by a reduction to the \textsc{supply-demand tree partition}.

\section{Tree S-partition via Supply-demand Tree Partition}
\label{SupplyDemand_section}

In this section, we present a linear time reduction from \textsc{tree S-partition} to \textsc{supply-demand tree partition}.  \textsc{supply-demand tree partition} can be solved in linear time \cite{PartitioningTreesSupplyDemand}, and the composition of all these linear time procedures results in a linear time algorithm for \textsc{tree $t$-spanner} in outerplanar graphs.

Let $\langle T, S \subseteq V(T), w , t \rangle$ be an input instance of \textsc{tree S-partition}.
  We now describe the construction of an instance $\langle T'$, $s: V_{s}(T') \to \mathbb{R}$,  $d:V_{d}(T') \to \mathbb{R} \rangle$ of \textsc{supply-demand tree partition}  from $\langle T, S \subseteq V(T), w, t \rangle$.
For every special vertex $u$ in $T$, add a new vertex $u'$ and a new edge $(u,u')$ to obtain $T'$.
The newly added vertices are considered as supply vertices $V_{s}$ and the rest of the vertices are considered as demand vertices $V_d$. Further, for each $u \in V_{s}$, assign $s(u) = t$ and for each $u \in V_{d}$, assign $d(u) = w(u)$. Clearly, the reduction takes $O(|V(T)|)$ time.

\begin{theorem}
\label{Theorem_2ndReduction}
$\langle T, S \subseteq V(T), w , t \rangle$ is a YES instance of \textsc{tree S-partition} if and only if $\langle T'$, $s: V_{s}(T') \to \mathbb{R}$,  $d:V_{d}(T') \to \mathbb{R} \rangle$ is a YES instance of \textsc{supply-demand tree partition}. 
 From a supply-demand tree partition of $T'$, a tree S-partition of  $T$ can be obtained in $O(|V(T)|)$ time.

\end{theorem}
\begin{proof}
In the forward direction, let $\pi = \{ V_{1}, \ldots, V_{|S|}\}$ be a tree S-partition of $T$.
For each $1 \leq i \leq |S|$, obtain $V'_i = V_i \cup \{ u' \mid u \in V_i \cap S \}$. Now, $\pi' = \{ V'_{1}, \ldots, V'_{|S|}\}$ is a supply-demand tree partition of $T'$.

In the reverse direction, let $\pi' = \{ V'_{1}, \ldots, V'_{k}\}$, where $k= |V_s|$ be an arbitrary supply-demand tree partition of $T'$. 
We first present the steps to construct a tree S-partition $\pi$ of $T$ such that $cost(\pi) \leq t$, followed by the correctness of the steps.

\textbf{Step 1:} For each $i$, let $T'_i = T'[V_i]$; Make the tree $T'_i$ rooted by considering the supply vertex in $T'_i$ as a root vertex.

\textbf{Step 2:} For each $i$, if $T'_i$ has exactly one vertex then do the following: Let $u'$ be the supply vertex in $T'_i$ and its neighbour  $u$ in $T'$ be present in $T'_j$. Delete the edge between $u$ and its parent in $T'_j$ and make $u'$ as the parent for $u$ by adding an edge between $u$ and $u'$.

\textbf{Step 3:} For each $i$, let $V'_i = V(T'_i)$. For each $i$, obtain $V_i$ by removing the supply vertex from $V'_i$. 
$\pi \leftarrow \{ V_1, \ldots, V_k \}$.

We now show that $\pi$ is a tree S-partition of $T$ such that $cost(\pi) \leq t$.
Step 1 is clear, as each $V'_i$ contains exactly one supply vertex. 
The effect of Step 2 is equivalent to moving the subtree rooted under $u$ including $u$ from $T'_j$ to $T'_i$. 
Thus Step 2 is clear.
Observe that the parent of $u$ in $T'_j$ is a not a supply vertex, 
because for each $v  \in V(T)$, at most one supply vertex is a neighbour of $v$. 
It implies that, at the end of Step 2, each $T'_i$ contains exactly one supply vertex and at least one demand vertex. 
Thus Step 3 is clear.
Further,  total demand in $V'_i$ as well as in $V'_j$ is at most $t$. 
Thereby $\{V'_1, \ldots, V'_k \}$ is a supply-demand tree partition of $T'$.
It follows that, the partition $\pi$ is a tree S-partition of  $T$ such that  $cost(\pi) \leq t$.

As we go through the edges in $T'$ a constant number of times and $|V(T')| \leq 2*|V(T)|$,   Steps 1-3 take $O(|V(T)|)$ time.
\qed
\end{proof}

\noindent \textbf{Proof of Theorem \ref{treeSpannerinOP}}
\begin{proof}
Let $\langle G,t \rangle$ be an instance of \textsc{tree $t$-spanner}  in outerplanar graphs.
We apply the linear-time transformation described in Section \ref{OuterPlanarreductionSubSection} on $\langle G,t \rangle$ to obtain an instance $\langle \tilde{T},w,S,t-1 \rangle $ of \textsc{tree S-partition}.
We then apply linear-time transformation described in Section \ref{SupplyDemand_section} on $\langle \tilde{T},w,S,t-1 \rangle $ to obtain an instance $\langle T',s,d \rangle$ of the \textsc{supply-demand tree partition}.
By using the linear-time algorithm in \cite{PartitioningTreesSupplyDemand}, we can determine whether $\langle T',s,d \rangle$ is an YES instance or NO instance.
If it is a NO instance, we can declare that there is no tree $t$-spanner in $G$.
 If it is an YES instance, then the algorithm in \cite{PartitioningTreesSupplyDemand} also produces a supply-demand tree partition $\pi'$. 
From $\pi'$, the construction of a tree S-partition $\pi$ of $\tilde{T}$ such that $cost(\pi) \leq t-1$  takes $O(|V(\tilde{T})|)$ time (Theorem \ref{Theorem_ConstructionFromSpartition}). 
From $\pi$, a tree $t$-spanner of $G$ can be obtained in $O(n)$ time (Theorem \ref{Theorem_2ndReduction}).
 As $V(\tilde{T}) = O(n)$, the existence of a tree $t$-spanner can be found in $O(n)$-time.
 Further, computing a tree $t$-spanner of $G$, if one exists, takes $O(n)$ time.
Thus the theorem.
\qed
\end{proof}



\bibliographystyle{plain}
\bibliography{references1}



\end{document}